\documentclass[11pt]{article}
\usepackage{amsthm}
\usepackage{graphicx}

\newtheorem{theorem}{Theorem}
\newtheorem{lemma}{Lemma}
\newtheorem{proposition}{Proposition}

\def\I_f{I_{\mbox{{\it\scriptsize for}}}}

\begin{document}

\title{A Note on the Inversion Complexity of Boolean Functions\\
       in Boolean Formulas}

\author{Hiroki Morizumi \vspace{2mm}\\
{\small Graduate School of Information Sciences, Tohoku University,}\\
{\small Sendai 980-8579, Japan}\\
{\small morizumi@ecei.tohoku.ac.jp}}

\date{}

\maketitle

\begin{abstract}
In this note, we consider the minimum number of NOT operators
in a Boolean formula representing a Boolean function.
In circuit complexity theory, the minimum number of NOT gates
in a Boolean circuit computing a Boolean function $f$ is called
the inversion complexity of $f$.
In 1958, Markov determined the inversion complexity of every Boolean
function and particularly proved that $\lceil \log_2(n+1) \rceil$
NOT gates are sufficient to compute any Boolean function on $n$ variables.
As far as we know, no result is known for inversion complexity
in Boolean formulas, i.e., the minimum number of NOT operators
in a Boolean formula representing a Boolean function.
The aim of this note is showing that we can determine the inversion complexity
of every Boolean function in Boolean formulas by arguments based on
the study of circuit complexity.
\end{abstract}

\section{Introduction} \label{sec:intro}

When we consider Boolean circuits with a limited number of NOT gates,
there is a basic question: Can a given Boolean function be computed by
a circuit with a limited number of NOT gates?
This question has been answered by Markov \cite{M58} in 1958.
The {\it inversion complexity} of a Boolean function $f$ is
the minimum number of NOT gates required to construct a Boolean circuit
computing $f$, and Markov completely determined the inversion complexity
of every Boolean function $f$.
In particular, it has been shown that
$\lceil \log_2(n+1) \rceil$ NOT gates are sufficient to compute
any Boolean function.

After more than 30 years from the result of Markov,
Santha and Wilson \cite{SW93} investigated the inversion complexity
in {\it constant depth circuits} and showed that on the restriction
$\lceil \log_2(n+1) \rceil$ NOT gates are not sufficient to compute
a Boolean function.
The result has been extended to {\it bounded depth circuits}
by Sung and Tanaka \cite{ST03}.
Recently we completely determined the inversion complexity
of every Boolean function in {\it non-deterministic circuits},
and particularly proved that one NOT gate is sufficient to compute
any Boolean function if we can use an arbitrary number of guess inputs
\cite{Morizumi}.

A Boolean circuit whose gates have fan-out one is called a {\it formula}.
Formulas are one of well-studied circuit models in circuit complexity
theory.
Note that a Boolean circuit whose gates have fan-out one corresponds
to a Boolean formula.
In this note, we investigate the inversion complexity in formulas,
which corresponds to the minimum number of NOT operators in a Boolean
formula representing a Boolean function.
As far as we know, there is no result for the inversion complexity
in formulas.
We completely determine the inversion complexity of every Boolean function
in formulas.

\section{Preliminaries} \label{sec:pre}

A {\it circuit} is an acyclic Boolean circuit which consists of AND gates of
fan-in two, OR gates of fan-in two and NOT gates.
A {\it formula} is a circuit whose gates have fan-out one.
We denote the number of NOT gates in a formula $C$ by $not(C)$.

Let $x$ and $x'$ be Boolean vectors in $\{0,1\}^n$.
$x \leq x'$ means $x_i \leq x'_i$ for all $1 \leq i \leq n$.
$x < x'$ means $x \leq x'$ and $x_i < x'_i$ for some $i$.
A Boolean function $f$ is called {\it monotone} if $f(x) \leq f(x')$
whenever $x \leq x'$.

A {\it chain} is an increasing sequence $x^1 < x^2 < \cdots < x^k$ of
Boolean vectors in $\{0,1\}^n$.
The {\it decrease} $d_X(f)$ of a Boolean function $f$ on a chain $X$
is the number of indices $i$ such that $f(x^i) \not\leq f(x^{i+1})$
for $1 \leq i \leq k-1$.
The {\it decrease} $d(f)$ of $f$ is the maximum of $d_X(f)$ over all
increasing sequences $X$.
We denote the inversion complexity of a Boolean function $f$
in circuits by $I(f)$.
Markov gave the tight bound of $I(f)$ for every Boolean function $f$.

\begin{proposition}[Markov\cite{M58}] \label{prop:markov}
For every Boolean function $f$,
$$I(f) = \lceil \log_2(d(f)+1) \rceil.$$
\end{proposition}

\section{Inversion Complexity in formulas} \label{sec:main}

\subsection{Result}

We denote by $\I_f(f)$ the inversion complexity of a Boolean function $f$
in formulas.
We consider only single-output Boolean functions.
The result of this note is the following one.

\begin{theorem} \label{thrm:main}
For every Boolean function $f$,
$$\I_f(f) = d(f).$$
\end{theorem}

\noindent In the rest, we prove Theorem~\ref{thrm:main}.

\subsection{Upper bound}

We prove $\I_f(f) \leq d(f)$.
We use a similar argument to one which is used to prove
Proposition~\ref{prop:markov} \cite{FTR}.

\begin{proof}[Proof (the upper bound of $\I_f(f)$)]
We use induction on $d(f)$.

\smallskip
\noindent {\it Base:}
$d(f)=0$. Then $f$ is monotone and $\I_f(f)=0$.

\smallskip
\noindent {\it Induction Step:}
Suppose
$$\I_f(f') \leq d(f')$$
for every Boolean function $f'$ such that $d(f') \leq d(f) - 1$.

\begin{figure}[t]
 \center{
  \includegraphics[scale=0.6]{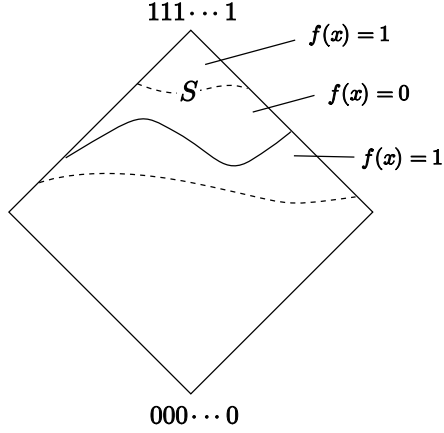}
  \caption{The separation of $f$.}
  \label{fig:func}
 }
\end{figure}

First we separate $f$ to two functions $f_1$ and $f_2$ as follows.
See Fig. \ref{fig:func}.
Let $S$ be the set of all vectors $x \in \{0,1\}^n$ such that for
every chain $X$ starting with $x$, $d_X(f)=0$.
We define $f_1$ and $f_2$ as follows:
$$f_1(x) = \cases{
            f(x) & \mbox{if $x \in S$;} \cr
            0    & \mbox{otherwise,}
           }$$
and
$$f_2(x) = \cases{
            1    & \mbox{if $x \in S$;} \cr
            f(x) & \mbox{otherwise.}
           }$$
We define $f_t$ as follows:
$$f_t(x) = \cases{
            1 & \mbox{if $x \in S$;} \cr
            0 & \mbox{otherwise.}
           }$$

By the definitions of $f_1$ and $S$,
\begin{equation}
d(f_1) = 0. \label{eq:f1}
\end{equation}
Next we show that $d(f_t) = 0$.
Let $x$ and $x'$ be Boolean vectors in $\{0,1\}^n$ such that $x \leq x'$.
Suppose that $f_t(x')=0$, i.e.,  $x' \not\in S$.
Since $x' \not\in S$, there is a chain $X'$ starting with $x'$
and such that $d_{X'}(f) \geq 1$.
Then for a chain $X$ which starts with $x$ and includes $X'$, $d_X(f) \geq 1$.
Therefore $x \not\in S$.
Thus if $f_t(x')=0$, then $f_t(x)=0$, which means
\begin{equation}
d(f_t) = 0. \label{eq:ft}
\end{equation}
Finally we show that
\begin{equation}
d(f_2) \leq d(f) - 1. \label{eq:f2}
\end{equation}
We assume that $d(f_2) > d(f) - 1$.
Since $f_2(x)=1$ for $x \in S$, there is a chain $X_1$ ending
in a vector $x' \not\in S$ and such that
$d_{X_1}(f) = d_{X_1}(f_2) > d(f) - 1$.
Since the $x'$ is not in $S$, there is a chain $X_2$ starting with $x'$ and
such that $d_{X_2}(f) \geq 1$.
Let $X'$ be the chain which is obtained by connecting $X_1$ and $X_2$.
Then,
\begin{eqnarray*}
d_X(f) & = & d_{X_1}(f) + d_{X_2}(f) \\
       & > & (d(f) - 1) + 1 = d(f).
\end{eqnarray*}
Thus a contradiction happens.

By the supposition and Eq.~(\ref{eq:f1}) to (\ref{eq:f2}),
there are a formula $C_2$ computing $f_2$ such that
$not(C_2) \leq d(f_2)$ and formulas $C_1$ and $C_t$ computing $f_1$ and $f_t$
respectively such that $not(C_1) = not(C_t) = 0$.
We construct a formula $C$ computing $f$ from $C_1, C_2$ and $C_t$
as $C$ computes the following:
$$f_1 \lor (f_2 \land \neg f_t).$$

The number of NOT gates in $C$ is
\begin{eqnarray*}
not(C) &  =   & not(C_1) + not(C_2) + not(C_t) + 1 \\
       & \leq & d(f_2) + 1 \\
       & \leq & d(f)
\end{eqnarray*}

\noindent We show that $C$ computes $f$ for each of the following two cases.

\medskip
\noindent Case 1: The input $x$ is in $S$.

Then $f_1(x) = f(x)$ and $f_t(x)=1$. Therefore
$$f_1 \lor (f_2 \land \neg f_t) = f_1 = f.$$

\medskip
\noindent Case 2: The input $x$ is not in $S$.

Then $f_1(x) = 0$, $f_2(x) = f(x)$ and $f_t(x)=0$. Therefore
$$f_1 \lor (f_2 \land \neg f_t) = f_2 = f.$$

\noindent Thus the formula $C$ computes $f$ and has at most $d(f)$
NOT gates. Therefore $\I_f(f) \leq d(f)$.

\end{proof}

\subsection{Lower bound}

We prove $\I_f(f) \geq d(f)$.
If the input of a NOT gate $N$ is $0$ and the output is $1$,
then we call the state of $N$ {\it up}.
If otherwise, we call the state {\it down}.
We denote by $not_{d}(C, x)$ the number of NOT gates 
whose states are down in a formula $C$ given $x$ as the input of $C$.

\begin{proof}[Proof (the lower bound of $\I_f(f)$)]
Let $C$ be a formula computing $f$.
Let $X$ be an increasing sequence $x^1 < x^2 < \cdots < x^k$ of
Boolean vectors in $\{0,1\}^n$ such that $d_X(f) = d(f)$.

\begin{lemma} \label{lem:lb1}
Let $x$ and $x'$ be Boolean vectors in $\{0,1\}^n$ such that $x < x'$,
$f(x)=1$ and $f(x')=0$. Then,
$$not_d(C, x') - not_d(C, x) \geq 1.$$
\end{lemma}

\begin{proof}
We change the input of $C$ from $x$ to $x'$.
Let $N_1, N_2, \ldots, N_m$ be all NOT gates which change from
down state to up state at the time.
Since $x < x'$, each $N_i$ for $1 \leq i \leq m$ is connected
from $N'_i$ which changes from
up state to down state by a path including no NOT gate.
Since the output of $C$ changes from $1$ to $0$, the output of $C$ is also
connected from $N'_o$ which changes from up state to down state
by a path including no NOT gate.
$N'_1, N'_2, \ldots, N'_m$ and $N'_o$ are distinguished from each other,
since $C$ is a formula.
Thus the number of NOT gates whose states are down increases
by at least one.
\end{proof}

\begin{lemma} \label{lem:lb2}
Let $x$ and $x'$ be Boolean vectors in $\{0,1\}^n$ such that $x < x'$.
Then,
$$not_d(C, x') - not_d(C, x) \geq 0.$$
\end{lemma}

\begin{proof}
We can use a similar argument to one of Lemma~\ref{lem:lb1}.
In this case, we do not consider $N'_o$.
\end{proof}

Since on $X$ the number of indices $i$ such that $f(x^i)=1$
and $f(x^{i+1})=0$ is at least $d(f)$,
by Lemma~\ref{lem:lb1} and \ref{lem:lb2},
$$not_d(C, x^k) - not_d(C, x^1) \geq d(f).$$
Thus $C$ includes at least $d(f)$ NOT gates.
\end{proof}


\begin{thebibliography}{99}

\bibitem{FTR}
M.~Fischer,
\newblock Lectures on network complexity,
\newblock {\itshape Technical Report} 1104,
         CS Department, Yale University, 1974 (revised 1996).
\bibitem{M58}
A.~A.~Markov,
\newblock On the inversion complexity of a system of functions,
\newblock {\itshape J. ACM} 5(4), pp.~331--334, 1958.
\bibitem{Morizumi}
H.~Morizumi,
\newblock Limiting negations in non-deterministic circuits,
\newblock {\itshape manuscript}, 2008.
\bibitem{SW93}
M.~Santha and C.~Wilson,
\newblock Limiting negations in constant depth circuits,
\newblock {\itshape SIAM J. Comput.} 22(2), pp.~294--302, 1993.
\bibitem{ST03}
S.~Sung and K.~Tanaka,
\newblock Limiting negations in bounded-depth circuits: an extension of Markov's theorem,
\newblock {\itshape Proc.~of 14th ISAAC}, LNCS vol.~2906, pp.~108--116, 2003.

\end{thebibliography}
\end{document}